%% file: main_arxiv.tex
\documentclass[sigconf,nonacm]{aamas}
\usepackage{balance} % for balancing columns on the final page

\usepackage{subcaption}
\usepackage{graphicx} % Required for inserting images

\usepackage{amssymb}
\usepackage{amssymb}
\usepackage{amsthm}
\usepackage{pgfplots}
\usepackage{array}
\usepackage{comment} % Add this in the preamble
\usepackage{tabularx}
\usepackage{enumitem}
\usepackage{algorithm,algpseudocode}
\usepackage{mathtools}
\usepackage{booktabs}
\usepackage{bbm}

\usepgfplotslibrary{fillbetween}
\pgfplotsset{compat=1.18} 
\usetikzlibrary{backgrounds} % for rectangle
\usetikzlibrary{calc, positioning, arrows.meta,positioning,decorations.pathreplacing}

\usetikzlibrary{external}
\setcopyright{ifaamas}
\acmConference[AAMAS '26]{Proc.\@ of the 25th International Conference
on Autonomous Agents and Multiagent Systems (AAMAS 2026)}{May 25 -- 29, 2026}
{Paphos, Cyprus}{C.~Amato, L.~Dennis, V.~Mascardi, J.~Thangarajah (eds.)}
\copyrightyear{2026}
\acmYear{2026}
\acmDOI{}
\acmPrice{}
\acmISBN{}

\acmSubmissionID{1019}

\title[AAMAS-2026 Formatting Instructions]{The Complexity of Strategic Behavior in Primary Elections}

\author{Colin Cleveland}
\affiliation{
  \institution{King's College London}
  \city{London}
  \country{United Kingdom}}
\email{colin.cleveland@kcl.ac.uk}

\author{Bart de Keijzer}
\affiliation{
  \institution{King's College London}
  \city{London}
  \country{United Kingdom}}
\email{bart.de_keijzer@kcl.ac.uk}

\author{Maria Polukarov}
\affiliation{
  \institution{King's College London}
  \city{London}
  \country{United Kingdom}}
\email{maria.polukarov@kcl.ac.uk}

\begin{abstract}

We study the computational complexity of strategic behaviour in primary elections. Unlike direct voting systems, primaries introduce a multi-stage process in which voters first influence intra-party nominees before a general election determines the final winner. While previous work has evaluated primaries via welfare distortion, we instead examine their game-theoretic properties. We formalise a model of primaries under first-past-the-post with fixed tie-breaking and analyse voters' strategic behaviour. We show that determining whether a pure Nash equilibrium exists is $\Sigma_2^{\mathbf P}$-complete, computing a best response is NP-complete, and deciding the existence of subgame-perfect equilibria in sequential primaries is PSPACE-complete. These results reveal that primaries fundamentally increase the computational difficulty of strategic reasoning, situating them as a rich source of complexity-theoretic challenges within computational social choice.

\end{abstract}

\keywords{Computational Social Choice; Algorithmic Game Theory; Strategic voting; Primary elections; PSPACE-completeness}

\ccsdesc[500]{Theory of computation~Algorithmic game theory}
\ccsdesc[500]{Theory of computation~Representations of games and their complexity}
\ccsdesc[300]{Theory of computation~Problems, reductions and completeness}
\newcommand{\BibTeX}{\rm B\kern-.05em{\sc i\kern-.025em b}\kern-.08em\TeX}
\newtheorem{theorem}{Theorem}

\newtheorem{proposition}[theorem]{Proposition}
\newtheorem{example}{Example}[section]
\newtheorem{remark}{Remark}
\newtheorem{definition}{Definition}

\begin{document}

\pagestyle{fancy}
\fancyhead{}
%%%%%%%%%%%%%%%%%%%%%%%%%%%%%For the anonymous 
%\author{Colin Cleveland\inst{1}\orcidID{0000-0003-0909-7518} \and Bart de Keijzer\inst{1}\orcidID{0000-0001-9465-0837} \and Maria Polukarov\inst{1}\orcidID{0000-0002-7421-3012}}
%
%\authorrunning{C. Cleveland et al.}
%\institute{King's College London, UK
%\email{ {colin.cleveland,bart.de\_keijzer,maria.polukarov\}@kxl.ac.uk}\\ \url{https://www.kcl.ac.uk/informatics} }
%
%%%%%Ende of for anonymous
\maketitle              % typeset the header of the contribution
%
%
%\linenumbers

\input{content}

\section*{Acknowledgements} 
BdK was supported by EPSRC grant EP/X021696/1. MP was supported by EPSRC grant EP/X038351/1. The authors thank the anonymous reviewers for %their 
constructive feedback and 
helpful suggestions.

\balance
\bibliographystyle{ACM-Reference-Format} 
\bibliography{references.bib}

\clearpage

\appendix
\input{appendix}

\end{document}

%% file: content.tex
\newcommand{\party}{\Pi}
\newcommand{\open}{\mathrm{open}}
\newcommand{\closed}{\mathrm{closed}}
\newcommand{\semi}{\mathrm{semi}}
\newcommand{\ind}{\mathrm{ind}}
\newcommand{\Pri}{\Pr\nolimits_i}
\newcommand{\eligible}{\kappa}
\newcommand{\ttutil}{\mathsf{U}}
\newcommand{\vo}{v_{\textbf{odd}}}
\newcommand{\ve}{v_{\textbf{even}}}
\newcommand{\Aodd}{A_{\textbf{odd}}}
\newcommand{\Aeven}{A_{\textbf{even}}}

\section{Introduction}
Democracy is often understood at the level of \emph{electorate $\rightarrow$ government}, yet an equally important question arises within political parties: how should parties select their own candidates?  
Historically, in many systems party elites simply chose nominees, with primaries serving, at best, as informal opinion polls.  
The 1968 U.S.\ Democratic convention, where Hubert Humphrey obtained the nomination without entering any primary, triggered widespread protests and prompted the institutionalisation of primary elections in the United States \citep{polsby1983consequences,cohen2008party}.  
Similar reforms have since spread worldwide: in some countries (e.g.\ the United States, Chile, Argentina) primaries are regulated by the state \citep{carey2003primarias,siavelis2009chile}, while in others (e.g.\ France, Italy) they remain intra-party procedures run by party organizations \citep{lefebvre2017_france_primaries,bernardi2013membership,sb2013italy_primaries}.  
Even in the United Kingdom, where primaries are not formally adopted, major parties increasingly allow supporters to participate in leadership selection \citep{quinn2024_parliamentarians_vs_members}.

Primary elections appear in several institutional variants—\emph{open}, \emph{closed}, \emph{semi-open}, or hybrid systems such as \emph{poll-weighted primaries}, in which opinion surveys influence results.  
These hybrids are especially common in newer democratic countries such as Taiwan and South Korea, where weak party membership structures leads parties to integrate public polling into nomination procedures \citep{hwang2017korea,reidhead2020_intraparty_sk_tw,lee2022presidential}.

Primaries are not only decision procedures but also communication channels between parties and their supporters.  
Yet, they invite strategic behaviour.  
For instance, in Colorado’s 2022 Republican primary, thousands of Democrats reportedly re-registered as unaffiliated voters to oppose Representative Lauren Boebert, ultimately without affecting the outcome \citep{cpr2022boebert}.  
Such ``crossover'' behaviour illustrates that additional choice can create incentives for manipulation and unintended losses.

From the party’s perspective, primaries may fail to maximize general-election performance, as polarizing candidates can prevail.  
From the voter’s perspective, they offer opportunities to reshape competition, but at non-trivial computational and informational costs.  
These observations motivate a formal analysis of the strategic structure of primaries.

%\medskip
\noindent\textbf{Our contribution.}  
We study primary elections through the lens of computational social choice and algorithmic game theory.  
Specifically, we:
\begin{itemize}
    \item Formalize a two-stage election model in which voters’ primary and general-election choices constitute explicit strategies (Section~\ref{sec:framework});
    \item Characterize the computational complexity of computing best responses, verifying, and deciding the existence of equilibria (Section~\ref{sec:BSEQ});
    \item Extend the framework to sequential multi-stage primaries, showing how temporal conditioning increases strategic complexity (Section~\ref{sec:sequential-BS}).
\end{itemize}

Our results show that multi-stage primaries amplify the reasoning burden faced by strategic agents: best-response computation is NP-complete, equilibrium existence rises to $\Sigma_2^{\mathbf P}$-completeness, and sequential conditioning yields PSPACE-completeness.  
Beyond elections, these results connect to broader questions in multi-agent reasoning—how decomposing a collective decision into stages expands the logical depth of equilibrium analysis within the polynomial hierarchy.

\paragraph{Complexity overview.}
We use the complexity classes NP, $\Sigma_2^{\mathbf P}$, and PSPACE.
NP consists of decision problems of the form $\exists x\,\varphi(x)$ with $\varphi$ computable in polynomial time.
The class $\Sigma_2^{\mathbf P}$ consists of problems of the form $\exists x\,\forall y\,\varphi(x,y)$, with $\varphi$ computable in polynomial time.
PSPACE consists of decision problems decidable in polynomial space; equivalently, it consists of problems expressible as $Q_1 x_1 \cdots Q_{p(n)} x_{p(n)}\,\varphi(x_1,\dots,x_{p(n)})$ with polynomially many alternating quantifiers, as in quantified Boolean formulas (\emph{QBF}).
Increasing the number of electoral stages corresponds to increasing quantifier alternation depth.

\section{Related Work}

Most research in computational social choice has focused on \emph{direct elections}, where all candidates compete simultaneously.  
\citet{borodin2024primaries} provide the first formal analysis of \emph{primary elections} in this domain, comparing multi-stage nomination processes with direct elections under various rules.  
Their results demonstrate that primaries can diverge substantially from direct elections both in theory and in empirical distortion analyses.  
Our work extends this direction by examining the \emph{strategic and computational} dimensions of voter behaviour in primaries rather than their welfare implications.

A complementary line of work studies \emph{strategic candidacy}, where candidates decide whether to enter or withdraw.  
This literature, initiated by \citet{dutta2001strategic,dutta2002equilibria} and later developed by \citet{brill2015strategic} and \citet{polukarov2015equilibria}, characterises equilibria over candidate participation decisions.  
Related work on \emph{electoral control}---for example, adding or deleting candidates---originates with \citet{bartholdi1989computational}.  
\citet{harrenstein2021hotelling} unify strategic candidacy and spatial competition in a Hotelling–Downs framework where parties nominate candidates, showing complexity bounds for equilibrium existence and links to facility-location and Voronoi games.  
More recently, \citet{schlotter2025nomination} analyse the \emph{Possible President} problem under Condorcet-consistent rules, exploring computational and parameterized complexity when nominees depend on rival parties’ choices.  
In contrast to these \emph{candidate-centric} perspectives, our focus is on the strategic computation of \emph{voters} within primary systems.

Another relevant stream concerns \emph{iterative voting}, where voters repeatedly adjust their ballots in response to others.  
\citet{meir2010convergence} initiated the study of convergence to equilibrium under plurality, with subsequent extensions by \citet{reijngoud2012plurality}, \citet{obraztsova2013plurality}, and \citet{lev2012convergence}.  
Surveys by \citet{meir2017iterative} and \citet{grandi2019iterative} synthesise this literature, which typically assumes a single-stage election and investigates behavioural convergence.  
Our analysis differs in treating voters as computationally bounded agents: we show that computing a best response is NP-complete once primaries precede the general election, and that equilibrium existence climbs to $\Sigma_2^{\mathbf P}$- or PSPACE-completeness as sequential conditioning is introduced.  
This places primary elections within a broader conversation about the complexity of equilibrium reasoning in multi-agent systems.

Finally, beyond computational studies, the design and consequences of primary elections have been widely examined in the social sciences.  
\citet{kenig2009selecting} surveys candidate-selection methods across parliamentary democracies, while \citet{cohen2008party} and \citet{sides2020representativeness} analyse elite influence and representativeness in U.S.\ primaries.  
These empirical insights complement the computational models developed in AI and game theory, together offering a more complete view of how institutional design shapes strategic reasoning in multi-agent collectives.

\section{Preliminaries}

For $n \in \mathbb{N}$, let $[n] = \{1,2,\dots,n\}$.  
Let $V = [n]$ be the \textit{set of voters} and $A = [m]$ the \textit{set of politicians}.\footnote{In a direct election, $A$ coincides with the candidate set.} 
We write $v_i$ for voter~$i \in V$ and $a_j$ for politician~$j \in A$.

Let $\Pi = [p]$ denote the \textit{set of parties}.  
Each party $k \in \Pi$ has a disjoint set of affiliated politicians $A_k \subseteq A$, so that $A_i \cap A_j = \emptyset$ whenever $i \neq j$.  
Hence every politician belongs to exactly one party.

After all primaries conclude, each party advances a single nominee to the \textit{general election (GE)}.  
The \textit{set of finalists} is thus a profile
\[
c = (c_1,\dots,c_p) \in C := A_1 \times A_2 \times \cdots \times A_p,
\]
where $c_k \in A_k$ is party~$k$’s nominee.

Each voter $i$ has a \textit{cardinal utility function} $u_i : A \to \mathbb{Q}$, where $u_i(a_j)$ is the utility from politician~$a_j$ winning the GE.  
Because $A$ is finite, we identify $u_i$ with its vector representation
\[
(u_{i,1},u_{i,2},\dots,u_{i,m}), \qquad u_{i,j}=u_i(a_j).
\]
The \textit{utility profile} $[u]$ collects all voters’ preferences and fully determines payoffs for every possible GE winner.

\subsection{Voting Framework}\label{sec:framework}

We model the election as a multiple-stage process: party primaries followed by the GE.  
Each party $k \in \Pi$ holds its primary in a fixed order (1,2,\dots,$p$).  
At each stage, voters cast their primary ballots simultaneously, and a unique nominee is declared immediately.  
Unless otherwise stated, both primaries and the GE use the \emph{first-past-the-post (FPTP)} rule: each voter casts one ballot, and the candidate with the most votes wins. Ties are broken according to a fixed, predetermined ordering of candidates (a \emph{fixed-tie breaking list}).

\paragraph{Voter strategies.}
A voter’s behaviour is described by a strategy
\[
s_i = \{b_i, g_i\},
\]
where:
\begin{itemize}
    \item $b_i = (b_{i,1},\dots,b_{i,p})$ is the \textbf{\textit{primary ballot vector}}.  
    Each coordinate $b_{i,k} \in A_k \cup \{0\}$ records either a vote for a politician in $A_k$ or abstention~($0$).
    \item $g_i : C \to A \cup \{0\}$ is the \textbf{\textit{general-election decision function}}.  
    For each finalist profile $c \in C$, it specifies the GE choice (or abstention) of voter~$i$.
\end{itemize}

This reduced form captures only the ex-post behaviour relevant to electoral outcomes, rather than a full extensive-form plan.  
A fully extensive description would need to specify actions for every possible primary path and finalist set—requiring $\Theta(m^{pn})$ space (see Appendix~\ref{appx:extensive-form}).  
The reduced form therefore trades completeness for analytical tractability.

\paragraph{Registration.}
Institutions vary in whether participation requires prior \emph{party registration}.  
In a \emph{closed} primary, each voter affiliates with one party before the season begins; if $i$ registers with party~$k$, only $b_{i,k}$ may be non-zero.  
In an \emph{open} primary, affiliation is not fixed in advance: at each stage~$k$, voter~$i$ decides whether to cast $b_{i,k}$, subject to the constraint that no voter participates twice in the same contest.  
Thus $b_i$ has the same domain in both regimes but different interpretations—static commitment versus contingent choice.  
Administrative details (deadlines, membership rolls, etc.) are abstracted away.

\paragraph{Participation regimes.}
A second institutional choice concerns how many primaries a voter may enter:
\begin{itemize}
    \item \textbf{Single-primary participation:} At most one coordinate of $b_i$ is non-zero.
    \item \textbf{Multiple-primary participation:} Each coordinate may be filled independently,  
    $b_i \in \prod_{k=1}^p (A_k \cup \{0\})$.
\end{itemize}
We adopt the multiple-participation regime as the general case; it subsumes the single-primary setting.

\paragraph{Access model for $g_i$.}
The general-election component $g_i$ may be exponentially large: since $|C| = \prod_{k=1}^p |A_k| = \Theta(m^p)$, an explicit table for $g_i$ would require exponential space.  
We therefore assume a standard RAM model in which $g_i$ is queried as an oracle with cost~$\tau=\tau(m,p)$.  
Our reductions use gadgets admitting constant-time queries ($\tau = O(1)$).  
We also allow \emph{succinct encodings}, where $g_i(c)$ is computed on demand.  
A canonical example is the \emph{listed-preference rule} $g_i(\cdot;L)$, where $L$ is an ordered list over $A\cup\{0\}$: given $c$, the voter selects the first element of $L$ appearing in $c$ (or abstains if $0$ precedes all finalists).  
This encoding uses only $O(m)$ space and supports $O(1)$ queries under the RAM model.

%\medskip
Collecting all voters’ strategies yields the profile $\mathsf{s}=(s_1,\dots,s_n)$.  
The institutional regime then determines the \textit{final winner} $\mathcal{W}(\mathsf{s}) \in A$.  
Throughout, voters are viewed as polynomial-time agents operating under these oracle assumptions.

\subsection{Voter Cost and Utility}

Casting a GE ballot incurs a participation cost $\kappa_i > 0$, and participation in party~$k$’s primary incurs $\kappa_i^k > 0$.  
Given a strategy $s_i$ and the strategies of others $\mathsf{s}_{-i} = (s_1,\dots,s_{i-1},s_{i+1},\dots,s_n)$, the \textit{total utility} of voter $v_i$ equals the utility derived from the eventual GE winner minus the participation costs incurred in the GE and any primaries. Formally,
\[
\ttutil_i(s_i,\mathsf{s}_{-i})
= u_i\!\bigl(\mathcal{W}(s_i,\mathsf{s}_{-i})\bigr)
- \mathbf{1}[g_i \neq 0]\,\kappa_i
- \sum_{k \in [p]} \mathbf{1}[b_{i,k} \neq 0]\,\kappa_i^k,
\]
where $\mathcal{W}(s_i,\mathsf{s}_{-i})$ denotes the GE winner.
Voters act to maximise their expected $\ttutil_i$ given beliefs about others’ strategies.

\subsubsection*{Toy Example: Two Parties with Single-Primary Participation}

Consider $\Pi=\{1,2\}$ with $A_1=\{a,b\}$ and $A_2=\{c,d\}$, and assume alphabetical tie-breaking.  
Politicians are positioned at $[10,6,4,0]$ on a line, with four voters located at $[0,4,6,10]$ and no participation costs.  
Utilities follow $u_i(a_j) = -|x_i - y_j|$ with $x_i$ and $y_j$ the location of voter $i$ and candidate $j$\footnote{Utilities may be shifted by an additive constant without affecting comparisons}.
Figure~\ref{fig:example_toy} illustrates the setting.

Each voter’s strategy is $s_i=\{b_i,g_i\}$, where $g_i$ is expressed as a ranked list:
\[
\begin{aligned}
s_1 &= \{[0,d],[d,c,b,a]\},\quad
s_2 = \{[0,c],[c,b,d,a]\},\\[2pt]
s_3 &= \{[b,0],[b,c,a,d]\},\quad
s_4 = \{[a,0],[a,b,c,d]\}.
\end{aligned}
\]
The finalists are $C=(a,c)$, and the GE winner is $c$, supported by $v_1,v_2,v_3$.  
Hence $[u_1(c),u_2(c),u_3(c),u_4(c)] = [-4,0,-2,-6]$.

\begin{figure}[ht]
\centering
    \includegraphics[width=0.9\linewidth]{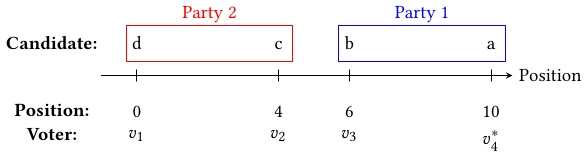}
\caption{Politician and voter positions in the toy example.}
\Description{
Spatial arrangement of politicians and voters in the toy example.
Party 1 (blue) fields politicians $a$ and $b$ at positions 10 and 6, while Party 2 (red) fields $c$ and $d$ at positions 4 and 0. 
Voters $v_1,v_2,v_3,v_4$ are placed correspondingly along the same line. Each voter’s utility is the negative of the distance between her position and the elected candidate’s position, so proximity determines preference orderings. }
\label{fig:example_toy}
\end{figure}

Two strategic votings improve $v_4$’s utility:
\begin{itemize}
    \item \emph{Hostile crossover.}  
    $v_4$ changes $b_4$ in $s_4=(b_4,g_4)$ to $[0,d]$.  
    Then $d$ defeats $c$ in Party~2’s primary ($2{:}1$) and $b$ defeats $a$ in Party~1’s ($1{:}0$), giving $C^*=(b,d)$.  
    In the GE, $b$ beats $d$ ($3{:}1$), yielding $\ttutil_4(s_4^*,s_{-4})=u_4(b)=-4> -6$.
    \item \emph{Honest coordination.}  
    $v_4$ instead votes $b_4=[b,0]$, giving $C^*=(b,c)$.  
    The GE is a $2{:}2$ tie resolved in favour of $b$, again improving $v_4$’s utility to $-4$.
\end{itemize}

These examples show that sincere primary voting need not maximise individual utility.
Other manipulations, such as \emph{honest crossover}, where a voter supports their favourite politician in a rival party’s primary—may also shift the finalist set toward a more desirable overall outcome.

%===[Decision branching]===
%We can already make an analysis here, assuming voters have full information about other voters.
%(Q_1) The algorithm to find the BS.
%(Q_2) The algorithm to find the EQ.
%I guess Q_1 will be polynomial and Q_2 will be NP hard.

\section{Computational Questions Over Primaries}\label{sec:BSEQ}

The framework above gives rise to two central decision problems.

\begin{quote}
\textbf{Question~1 (Best Response).}\\
\emph{Instance:} Voter set $V$, utility profile $[u]$, politician sets $A_1,\dots,A_p$, tie-breaking list $\mathcal{L}$, and strategy profile $\mathsf{s}$.\\
\emph{Task:} For a given voter $i$, compute a strategy $s_i^*$ that maximizes 
$\ttutil_i(s_i^*, \mathsf{s}_{-i})$.
\end{quote}

\begin{quote}
\textbf{Question~2 (Equilibrium).}\\
\emph{Instance:} $V$, $[u]$, and politician sets $A_1,\dots,A_p$.\\
\emph{Task:} Decide whether there exists a strategy profile $\mathsf{s}^*$ such that, for every voter $i$ and every alternative strategy $s_i'$,
\[
\ttutil_i(\mathsf{s}^*) \;\geq\; \ttutil_i(s_i', \mathsf{s}^*_{-i}).
\]
\end{quote}

For clarity, we begin with the simplest timing structure in which all parties hold their primaries simultaneously.  
This reduces the election to two stages: primaries followed by the GE, while retaining the essential strategic components.  
In Section~\ref{sec:sequential-BS}, we relax this assumption and consider the sequential case, where primaries unfold in turn across parties.

\subsection{Best-Response Search Space and Running Time}\label{BR_Analysis}

A naïve way to compute the best response of voter~$i$ is to enumerate all possible strategies.  
Each strategy consists of (i) a primary ballot for every party and (ii) a general-election (GE) decision rule $g_i\!:\, C\!\to\! A\!\cup\!\{0\}$ that \emph{selects among finalists}.
This yields the upper bound
\[
\Bigl(\prod_{k=1}^p (1+|A_k|)\Bigr) \cdot (|A|+1)^{|C|},
\]
since each party ballot offers $|A_k|{+}1$ options (including abstention), and a GE rule must specify an action for each finalist profile, where $|C|=\prod_k m_k=O(m^p)$.  

However, when computing a best response, we never need to enumerate the entire $g_i$.  
Once the finalists $c\!\in\! C$ are known, voter~$i$ simply chooses greedily among $\{c_1,\dots,c_p,0\}$—the available finalists and abstention—thereby reducing the GE decision to $m{+}1$ evaluations.  
Hence the dominant factor in the search space comes from enumerating primary ballots \footnote{When the number of candidates, parties, or primary rounds is small, the question of optimal strategy may be encoded as Presburger arithmetic formulas and thus, may not need enumeration \cite{Koutecky_Talmon_2021}}.

%\smallskip

To analyse the running time, we decompose the computation into two components.

\paragraph{Baseline procedure.}
\begin{enumerate}
  \item \emph{Pre-computation.}  
  For each party~$k$, compute primary tallies from $\mathsf{s}_{-i}$, requiring $O(np)$ time by scanning all other voters once.
  \item \emph{Enumeration of $i$'s primary actions.}  
  For every combination $b_i\!\in\!\prod_k(A_k\cup\{0\})$:
  \begin{enumerate}
    \item Update tallies locally to reflect $b_i$, determine each party’s nominee under FPTP (with fixed tie-breaking), and obtain the finalist tuple $c_x\!\in\! C$.  
    This takes $O(mp)$ per $b_i$, using running arg-maxes for each party.
    \item \emph{Evaluate the GE outcome.}  
    For each $i'\!\neq\! i$, query $g_{i'}(c_x)$ (time $\tau$ per query) and increment that choice’s tally.  
    This step costs $O(n\tau + p)$ per $b_i$ ($p$ accounts for the final arg-max over GE tallies).  
    Voter~$i$ then greedily picks her GE action from $\{c_{x,1},\dots,c_{x,p},0\}$ and computes $\ttutil_i$.
  \end{enumerate}
\end{enumerate}

Aggregating these terms yields the following bound.

\begin{proposition}[Best-response running time]\label{prop:br-runtime}
Let $\tau$ be the time to evaluate a single voter’s GE decision $g_i(c)$ for a given finalist profile $c$.
For a fixed voter $i$, the enumeration algorithm over $b_i \in \prod_{k=1}^p (A_k \cup \{0\})$ runs in
\[
O\bigl(np + (m{+}1)^p (mp + n\tau)\bigr)
\]
For fixed $p$ and constant $\tau$,
\[
T_{\mathrm{BR}}
= O\!\bigl( n + m^{p}\,(mp + n) \bigr)
= O\!\bigl( m^{p+1} + n\,m^{p} \bigr),
\]
so best response is computable in time polynomial in $n$ and $m$.
\end{proposition}

\paragraph{Explicit vs.\ succinct GE behaviour.}
The size of the finalist set is $|C|=\prod_k|A_k|=O(m^p)$.  
Under an \emph{explicit} representation, each $g_i:C\!\to\!A\!\cup\!\{0\}$ is stored as a table of $\Theta(m^p)$ entries with $O(1)$ query time, implying that the input already encodes $\Theta(nm^p)$ information.  
Since $T_{\mathrm{BR}}=O((nm^p)^2)$, the best response remains polynomial in input length even if $p$ grows with~$m$.

By contrast, if $g_j$ is given \emph{succinctly} (e.g., the listed-preference presentation), the input size drops to $O(nm)$.  
The exponential factor $(m+1)^p$ in the action space then re-emerges, making complexity sensitive to the representation gap.

\begin{remark}[Polynomial time versus input size]\label{rem:inputsize}
For fixed~$p$ and succinctly evaluable GE rules, best response is polynomial in $n,m$.  
With explicit~$g$, the input size is already exponential in~$p$, so best response remains polynomial in input length.  
Only when $g$ is succinct do we encounter a genuine complexity increase.
\end{remark}

\paragraph{Special regimes.}
Two simplified variants are worth noting:
\begin{enumerate}
  \item \emph{Single-Participation primaries.}  
  If each voter can participate in at most one primary, the action space reduces to $O(m)$, and best response runs in $O(m^2 + mn)$.
  \item \emph{Fixed GE behaviour.}  
  If a voter’s GE choice is predetermined (e.g., always supporting a fixed party’s nominee), the GE stage adds no computational burden, and the runtime remains as stated in the proposition.
\end{enumerate}

\noindent\textbf{Towards NP-completeness.}  
When GE behaviours are given succinctly, a natural question arises: does a polynomial-time algorithm for best response exist?  
We show that the corresponding decision version is NP-complete.

\begin{definition}[Best-Response-at-Least-$U$ (BR-$\ge U$)]
Under FPTP rule, fix $m=|A|$, $p=|\Pi|$, and a profile $\mathsf{s}_{-i}$ of all voters except $i^\star$.  
The decision problem \textsc{BR-$\ge U$} is defined as follows:

\begin{quote}
\emph{Instance:} An election with parties $\Pi=\{1,\dots,p\}$ and politician sets $A_k$, a fixed tie-breaking list $\mathcal{L}$, other voters’ strategies $\mathsf{s}_{-i^\star}$, and succinct GE functions $g_j:C\to A\cup\{0\}$ (query-able in $O(1)$). 
A utility threshold $U\in\mathbb{Q}$.\\[2pt]
\emph{Question:} Does there exist a strategy $s_{i^\star}$ such that 
$\ttutil_{i^\star}(s_{i^\star},\mathsf{s}_{-i^\star}) \ge U$?
\end{quote}
\end{definition}

\begin{theorem}[NP-completeness of \textsc{BR-$\ge U$}]\label{thm:BR-NPC}
When the number of parties~$p$ is part of the input and GE behaviours are given succinctly, 
\textsc{BR-$\ge U$} is NP-complete with fixed tie-breaking, even with zero participation costs.
\end{theorem}

\begin{proof}
\emph{Membership.}  
A certificate consists of $i^\star$’s strategy.  
Verification simulates the primaries, evaluates the finalists, queries each $g_j(c)$, tallies GE votes, and computes $\ttutil_{i^\star}$—all in polynomial time under the RAM model.

%\smallskip
\emph{Hardness (from 3-SAT).}  
Given a Boolean formula $\Phi(x_1,\dots,x_p)$, we construct an election instance where the primary stage encodes variable assignments and the GE tallies evaluate clauses.

\paragraph{Politicians (anchors as a fixed party).}
For each variable $x_k$ create a party with politicians $\{x_k,\neg x_k\}$.  
Add two \emph{dummy} parties $A_{\top}=\{x_\top\}$ and $A_{\bot}=\{x_\bot\}$ that skip meaningful primaries and always advance their unique nominees to the GE (they contribute only a constant factor to $|C|$ and do not affect the analysis).
Thus each finalist profile has the form $c=(x_\top,x_\bot,c_1,\dots,c_p)$.

\paragraph{Primaries.}
Only $i^\star$ votes in the variable parties’ primaries.  
Deterministic tie-breaking makes her ballot pivotal, so she freely selects nominees $\{c_1,\dots,c_p\}$ encoding an assignment $\sigma$.

\paragraph{GE voters.}
Two fixed blocs of size $Q=q+1$ vote for $x_\top$ and $x_\bot$, respectively.  
For each clause $C_t=(\ell_{t1}\vee\ell_{t2}\vee\ell_{t3})$, add one clause voter with list-form GE rule
\[
g_t(c;L_t),\qquad L_t=[\ell_{t1},\ell_{t2},\ell_{t3},x_\bot,0],
\]
which selects the first finalist in $L_t$ (succinct, $O(1)$ query).  
The distinguished voter $i^\star$ votes for $x_\top$.

\paragraph{Correctness and size.}
If $\Phi$ is satisfiable, $i^\star$ picks nominees according to a satisfying assignment, every clause voter supports some literal finalist, and $x_\top$ beats $x_\bot$ by one vote ($q{+}2$ vs.\ $q{+}1$).  
If $\Phi$ is unsatisfiable, some clause is falsified and its voter supports $x_\bot$, which then ties or beats $x_\top$; fixed tie-breaking favours $x_\bot$.  
The construction uses $O(p)$ politicians and $O(q)$ voters, hence is polynomial.

%\smallskip
Thus \textsc{BR-$\ge U$} is NP-hard; combined with membership, it is NP-complete.
\end{proof}

With $Q=q+1$ fixed supporters each for $x_\top$ and $x_\bot$, and only $q$ clause voters in total, any non-anchor finalist can obtain at most $q$ votes—strictly fewer than $q+2$.  
Hence the final winner is always either $x_\top$ or $x_\bot$.

The reduction can be further simplified by removing the fixed blocs.  
Define each clause voter $v_t$ so that if at least one literal $\ell_{ti}$ of clause $C_t$ appears among the finalists, $v_t$ abstains; otherwise $v_t$ votes for $x_\bot$.  
This version still uses $O(1)$-query rules and needs only one fixed supporter for $x_\top$, but it is less behaviourally natural (a kind of “protest abstention”), so we retain the bloc-based construction.

%\medskip
Thus, in our formulation, best-response computation remains tractable only when numbers of primaries or explicit GE rules are fixed.

\subsection{Pure Equilibrium: Verification, Existence, and Running Time}\label{sec:pure-eq}

As in $\S$~\ref{BR_Analysis}, we assume GE behaviour is given succinctly, so each voter’s strategy is $s_i=(b_{i,1},\dots,b_{i,p},g_i)$ with $b_{i,k}\in A_k\cup\{0\}$ and $g_i:C\to A\cup\{0\}$ queryable in $O(1)$ time.

We study two problems: \textsc{NE-VERIFY}—checking whether a given profile is a pure Nash equilibrium — and \textsc{NE-EXIST} — deciding whether one exists.

\begin{definition}[Equilibrium Verification (\textsc{NE-VERIFY})]
Under FPTP, given parties $\Pi=\{1,\dots,p\}$ with politician sets $A_k$ ($|A_k|\!\ge\!1$), a fixed tie-breaking list $\mathcal{L}$, utilities $[u]$, succinct GE rules $(g_j)_j$, and a strategy profile $\mathsf{s}=(s_i)_{i\in V}$,  
decide whether for all voters $i$ and all alternative strategies $s'_i$,
\[
\ttutil_i(\mathsf{s}) \;\ge\; \ttutil_i(s'_i,\mathsf{s}_{-i}).
\]
\end{definition}

\paragraph{Trivial equilibria.}
Before analysing equilibrium existence, we exclude outcome-constant profiles where everyone unanimously supports the same politician across stages; the winner is fixed and no unilateral deviation can change it.  
We therefore focus on a small subset of strategic voters:
\[
V = V_{\mathrm{env}} \cup V_{\mathrm{strat}},
\]
where $V_{\mathrm{env}}$ are fixed environment voters and $V_{\mathrm{strat}}$ are those whose decisions affect the GE outcome.

\begin{definition}[Equilibrium Existence (\textsc{NE-EXIST})]
Given the same model, decide whether there exists a profile $\mathsf{s}^*=(s_i^*)_{i\in V}$ such that,  
for every $i\in V_{\mathrm{strat}}$ and every unilateral deviation $s'_i$,
\[
\ttutil_i(\mathsf{s}^*) \;\ge\; \ttutil_i(s'_i,\mathsf{s}^*_{-i}).
\]
\end{definition}

\subsubsection{Verification}

The link between \textsc{NE-VERIFY} and \textsc{BR-$\ge U$} is direct:  
\textsc{BR-$\ge U$} asks whether there exists a profitable deviation, whereas \textsc{NE-VERIFY} asks whether no such deviation exists for any voter.  
A certificate for “not NE’’ is a pair $(i,s'_i)$; verifying it recomputes the finalists under $(s'_i,\mathsf{s}_{-i})$, evaluates $i$’s new GE action $y=g'_i(c)$ (a constant choice suffices), queries all $g_j(c)$ for $j\neq i$, tallies, and compares $\ttutil_i$.  
All these operations are polynomial, so “not NE’’ lies in NP and hence \textsc{NE-VERIFY} is in coNP.

\begin{theorem}[coNP-completeness of NE-VERIFY]\label{thm:NE-verif}
Given a strategy profile $\mathsf{s}$, deciding whether $\mathsf{s}$ is a pure Nash equilibrium is \textbf{coNP}-complete under FPTP with fixed tie-breaking, even with zero participation costs and succinct GE evaluation.
\end{theorem}

\begin{proof}
\emph{Membership.}  
A witness for “not NE’’ is $(i,s'_i)$; verifying it takes polynomial time, so \textsc{NE-VERIFY}~$\in$~coNP.

%\smallskip
\emph{Hardness.}  
Reduce from \textsc{UNSAT} using the same construction as in Theorem~\ref{thm:BR-NPC}.  
Fix $\mathsf{s}$ to select all “false’’ nominees.  
If $\Phi$ is satisfiable, a pivotal voter can deviate to a satisfying assignment and improve—so $\mathsf{s}$ is not an NE.  
If $\Phi$ is unsatisfiable, no unilateral deviation can yield a satisfying assignment, so no voter can improve and $\mathsf{s}$ is an NE.  
Thus $\Phi\in\textsc{UNSAT} \iff \mathsf{s}\in\textsc{NE-VERIFY}$, giving coNP-hardness.
\end{proof}

\subsubsection{Existence}

\textsc{NE-EXIST} has the quantifier form
\[
\exists \mathsf{s}\;\forall i\;\neg(\exists s'_i\text{ profitable}),
\]
where the inner predicate (“is there a profitable deviation?”) is in NP;
hence \textsc{NE-EXIST}~$\in\Sigma_2^{\mathbf P}$.
A standard gadget-based reduction from $2$-\textsc{QBF} (i.e., deciding $\exists X\,\forall Y\,\varphi(X,Y)$)—using closed primaries so that $\vo$ controls $X$ and $\ve$ controls $Y$, and the same clause/rescue voter gadgets with succinct GE functions—yields $\Sigma_2^{\mathbf P}$-hardness.
Before the formal proof, we show why pure equilibria may fail to exist.

\begin{example}\label{EXM:NO-NE-Kappa}
\textbf{No pure NE with closed primaries.}
Consider FPTP with two parties, 
\(A_{\mathrm{odd}}=\{a_1,a_3\}\) and 
\(A_{\mathrm{even}}=\{a_2,a_4\}\).
Four environment voters \(v_1,\dots,v_4\) have cyclic GE lists
$(a_1,a_2,a_3,a_4)$, $(a_2,a_3,a_4,a_1)$, $(a_3,a_4,a_1,a_2)$, $(a_4,a_1,a_2,a_3)$,
and two primary voters:
\[
\vo:\ a_1 \sim a_3 \succ a_2 \sim a_4,\qquad
\ve:\ a_2 \sim a_4 \succ a_1 \sim a_3.
\]
Say no cross-party participation
(\(\kappa^{\mathrm{odd}}_{\mathrm{even}}=\kappa^{\mathrm{even}}_{\mathrm{odd}}=\infty\)),
so each primary has a single strategic voter.
The induced $2\times2$ game of nominees is
\[
\begin{array}{c|cc}
 & a_2 & a_4\\
\hline
a_1 & \text{odd wins} & \text{even wins}\\
a_3 & \text{even wins} & \text{odd wins}
\end{array}
\]
which is the \emph{matching-pennies} payoff structure:
no pure Nash equilibrium exists (the unique mixed equilibrium is uniform).
\end{example}

Even without participation costs, pure equilibria may fail to exist:  
cyclic tie-breaking across primaries ensures that some voter can always deviate profitably (see Appendix~\ref{appx:no-kappa} for the full construction).

\begin{proposition}\label{prop:noNE-nokappa}
Under FPTP with fixed tie-breaking, primary elections may admit no pure-strategy Nash equilibrium.
\end{proposition}

\begin{proof}
Examples~\ref{EXM:NO-NE-Kappa} and~\ref{EXM:NO-NE} exhibit deviation cycles—with and without participation costs—so at least one voter can always improve unilaterally; hence no pure NE exists.
\end{proof}

\begin{theorem}[$\Sigma_2^{\mathbf P}$-completeness of NE existence]\label{thm:NE-exist}
Under FPTP, deciding whether a pure Nash equilibrium exists is $\Sigma_2^{\mathbf P}$-complete with fixed tie-breaking, even when GE behaviour is finalist-based (list actions).
\end{theorem}

\begin{proof}
\emph{Membership.}
The problem has the form
$\exists \mathsf{s}\ \forall i\ \neg(\exists s'_i$ profitable).
Given $\mathsf{s}$ and $(i,s'_i)$ we can recompute finalists, evaluate list-based GE actions in $O(1)$ per voter, tally, and compare utilities in polynomial time, hence \textsc{NE-EXIST}$\in\Sigma_2^{\mathbf P}$.

%\smallskip
\emph{Hardness (from $2-$\textsc{QBF}).}
Let $\Phi(X,Y)=\exists X\,\forall Y:\varphi(X,Y)$ be a quantified 3-CNF.
We build a primary-election game $G_\Phi$ (FPTP, fixed tie-breaking) that admits a pure NE iff $\Phi$ is true.

\paragraph{Parties and players.}
Each variable $z\in X\cup Y$ corresponds to a party with nominees $\{z,\neg z\}$.
Add two top-level parties $\Aodd=\{a_1,a_3\}$ and $\Aeven=\{a_2,a_4\}$ with tie-breaking
$a_1\succ a_2\succ a_3\succ a_4$.
Two strategic voters act:
$\vo$ controls all $X$-parties and chooses $a_o\in\{a_1,a_3\}$,
$\ve$ controls all $Y$-parties and $a_e\in\{a_2,a_4\}$.
Both care only about whether their top-level party wins the GE.

\paragraph{Background and clauses.}
Give $\Aodd$ a baseline margin $+1$ (e.g., $r{+}1$ fixed odd supporters vs.\ $r$ even).
For each clause $C_t=(\ell_{t1}\vee\ell_{t2}\vee\ell_{t3})$ add four clause voters
$v_{t1},\dots,v_{t4}$ with lists
\[
L_{ti}=(\ell_{t1},\ell_{t2},\ell_{t3},a_i,a_{i+1},a_{i+2},a_{i+3})\quad(\text{mod }4).
\]
Each votes for the first finalist in $L_{ti}$.
If $C_t$ is satisfied, all four vote for literals; if falsified, they fallback cyclically to $\{a_1,a_2,a_3,a_4\}$, producing a $3{:}1$ split that favours one top-level party depending on $(a_o,a_e)$
(as in Example~\ref{EXM:NO-NE-Kappa}).

The construction uses $O(|X|{+}|Y|{+}r)$ politicians and voters, all with succinct GE rules, so its size is polynomial.

%\medskip
\noindent(\,$\Rightarrow$\,)
If $\Phi$ is true, let $x^\star$ satisfy $\forall Y:\varphi(x^\star,Y)$.
In profile $S^\star=(x^\star,a_1;\,y,a_e)$ for any $(y,a_e)$,
every clause is satisfied, so all clause voters support literals and the baseline $+1$ margin lets $\Aodd$ win.
Neither $\ve$ (who always loses) nor $\vo$ (already winning) can improve, so $S^\star$ is a pure NE.

%\smallskip
\noindent(\,$\Leftarrow$\,)
If $\Phi$ is false, then for every $x$ there exists $y$ falsifying at least one clause.
Each falsified clause triggers a $3{:}1$ fallback favouring one party;
switching the top-level nominee ($a_o$ or $a_e$) reverses that advantage.
Hence, after any $(x,a_o;y,a_e)$ the loser can deviate profitably,
and deviations alternate indefinitely—no pure NE exists.

%\smallskip
Thus $G_\Phi$ has a pure NE iff $\Phi$ is true, establishing
$\Sigma_2^{\mathbf P}$-hardness.  
Combined with membership, \textsc{NE-EXIST} is $\Sigma_2^{\mathbf P}$-complete.
\end{proof}

\noindent\textbf{Takeaway.}  
Equilibrium analysis in primary elections spans the second level of the polynomial hierarchy.  
Verification of a proposed profile lies in \textbf{coNP}, as disproving stability only requires exhibiting a profitable deviation; deciding existence introduces an additional existential quantifier over profiles, yielding \textbf{$\Sigma_2^{\mathbf P}$}-completeness.  
While best responses are tractable in isolation, the search for mutually stable outcomes is theoretically intractable.

\section{Complexity Under Sequential Primaries: Commitment vs.\ Conditioning}\label{sec:sequential-BS}

Our $\Sigma_2^{\mathbf P}$ result already assumes that voters’ full strategies are evaluated \emph{ex post}:  
each voter commits in advance to a complete plan across all stages, and equilibrium verification checks whether any unilateral deviation is profitable.

The picture changes once primaries hold sequentially, and later-stage voters may \emph{condition} their actions on earlier outcomes, as in subgame-perfect play.  
In Theorem~\ref{thm:NE-exist}, we used $\exists$ for the designer-controlled variables ($\vo$) and $\forall$ for the adversarial ones ($\ve$) to obtain the $\Sigma^{\mathbf P}_2$-complete result.  
Extending this reasoning, each sequential stage introduces an additional quantifier block of the same type.  
Thus a process with $p$ stages behaves like a $\Sigma^{\mathbf P}_{O(p)}$ problem, and when $p$ is unbounded, the alternation depth is unbounded, suggesting PSPACE-completeness by analogy with quantified Boolean formulas (QBF).  
We formalize this correspondence below.

\paragraph{Sequential primaries (model recap).}
There are $p$ parties indexed by $k \in [p]$, each with a nonempty politician set $A_k$, and let 
$C \coloneqq A_1 \times \cdots \times A_p$ denote the set of nominee tuples.  
Primaries occur in a fixed order $1,2,\dots,p$.  
At stage $t \in [p]$, a designated subset of strategic voters $ V_{\mathrm{strat}} \subseteq V$ acts, while all others follow fixed behaviours.  
Within each stage, moves are simultaneous.  
After resolving the party-$k$ primary (under FPTP with fixed tie-breaking), its nominee $c_k \in A_k$ is fixed, and the process proceeds to stage $k{+}1$.  
After all stages, the general-election (GE) outcome is computed from $c = (c_1,\dots,c_p) \in C$ via a polynomial-time evaluable rule $\mathcal{W} : C \to A$ (e.g., FPTP with fixed tie-breaking).

\paragraph{Information and histories.}
Let $C_{<k} \coloneqq A_1 \times \cdots \times A_{k-1}$ denote the space of nominee prefixes.  
Before acting at stage $k$, all players observe the realized prefix $h_k \in C_{<k}$, i.e., the sequence of winners $(c_1,\dots,c_{k-1})$.%
\footnote{If only partial or noisy information is available (e.g., observing only $c_{k-1}$), replace $C_{<k}$ by the corresponding information set $\mathcal{H}_k$.}
A ballot for party $k$ is an element of $A_k \cup \{0\}$.  
For voter $i$, a (behavioural) strategy at party $k$ is a function
\[
f^{p}_{i,k} : C_{<k} \longrightarrow A_k \cup \{0\}.
\]
Hence a \emph{history-dependent strategy} for voter $i$ is the tuple
\[
s_i = \bigl(f^{p}_{i,1}, \dots, f^{p}_{i,p},\; g_i\bigr),
\]
where $g_i$ represents the GE behaviour (e.g., a map $g_i : C \to A \cup \{0\}$ or a ranked-list policy, depending on the GE model).  
A full strategy profile is $\mathrm{s} = (s_i)_{i \in V}$.

Given $ \mathrm{s}$, the induced nominee tuple $c( \mathrm{s} ) \in C$ is generated by sequentially applying FPTP with fixed tie-breaking at each stage using ballots $\{ f^{p}_{i,k}(h_k) \}_{i \in V_k}$, where $h_k$ is the realized prefix.  
The GE winner is $\mathcal{W}\bigl(c( \mathrm{s} )\bigr)$.

%\begin{figure}[ht]
    %\centering
    %\resizebox{.8\linewidth}{!}{\input{Figures/Illustration_Completed_tree}}
    %\caption{Game tree of two voters in a sequential primary. Nodes labelled %\emph{1’s} and \emph{2’s} represent party~1’s and party~2’s primaries, respectively. At each node, both voters simultaneously choose ballots $(a_i,a_j)$.}
    %\Description{ The figure shows a horizontal game tree with three stages. The leftmost node, labelled “1’s”, represents Party 1’s primary election. From it, several arrows branch rightward, each labelled with a joint ballot such as $(a_1,a_1)$ or $(a_2, a_2)$. These arrows connect to three vertically stacked nodes, all labelled “2’s”, showing Party 2’s primary after different stage-1 outcomes.Between the 2’s nodes is a vertical ellipsis indicating additional possibilities. From the middle “2’s” node, a dashed continuation marked “...” leads rightward to a node labelled “GE” (general election).    The GE node branches into two terminal circles, labelled $(c_1,c_1)$ and $(c_p,c_p)$, with a vertical ellipsis between them to indicate further outcomes.The diagram thus depicts how the sequential primary process unfolds: Party 1 acts first, Party 2 responds based on observed history, and the process concludes in a general election.    }
    %\label{fig:tree}
%\end{figure}

\begin{figure}[ht]
    \centering
%    \resizebox{0.8\linewidth}{!}{\input{Figures/Extensive_Form_Sequential}}
    \includegraphics[width=0.8\linewidth]{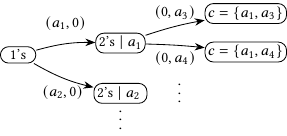}
    \caption{A partial game tree of two voters in a sequential primary. Nodes labelled ``1’s'' and ``2’s $\mid a_j$'' represent party~1’s primary and party~2’s primary conditional on candidate $a_j$ winning party~1’s primary, respectively. At each node, both voters simultaneously cast ballots $(a_i,a_{i'})$. Leaf nodes $c=\{c_1,c_2\}$ denote the general-election finalists.}
    
   \Description{
    The figure shows a partial horizontal game tree representing a two-voter sequential primary. The leftmost node, labelled \emph{1’s}, represents party~1’s primary election. Edges branching to the right correspond to joint ballot choices such as $(a_1,a_1)$ or $(a_2,a_2)$. These edges lead to nodes labelled \emph{2’s}\,$\mid a_j$, indicating party~2’s primary conditional on candidate $a_j$ winning party~1’s primary; a vertical ellipsis indicates additional possibilities. From each \emph{2’s}\,$\mid a_j$ node, the process proceeds to the GE, which branches into terminal nodes corresponding to finalist pairs.
    Overall, the diagram illustrates that party~1 acts first, party~2 responds after observing the stage~1 outcome, and the process concludes with a GE.}
    \label{fig:tree}
\end{figure}

The tree in Figure~\ref{fig:tree} illustrates the extensive form of a two-voter sequential primary.
Stage~1 (\emph{1’s}) and stage~2 (\emph{2’s}) correspond to the primaries of party~1 and party~2, respectively, where nodes labelled "$2\text{’s}\mid a_j$" indicate that candidate $a_j$ won party~1’s primary.
At each node, both voters simultaneously cast ballots $(a_i,a_{i'})$, and each branch represents a possible nominee outcome leading to the GE.

Note that the branching factor at each node of party~$k$’s primary is $(m_k+1)^n$, reflecting all possible combinations of $n$ voters’ simultaneous ballots over $m_k$ politicians (including abstention).
Each node can thus be indexed by an $n$-tuple $(a_{k1}, a_{k2}, \dots, a_{kn})$ representing the voters’ ballots. 
Although the number of nodes grows exponentially in $n$, the full tree can still be traversed using polynomial space.

\paragraph{Strategic scope and representation.}
We assume that only a small subset of voters are strategic, while all others are \emph{environmental voters} whose behaviour remains fixed as $(b_i, g_i)$ defined in $\S$~\ref{sec:framework}.  
If each function $f^{p}_{i,k}$ were represented explicitly, its table size would grow as $\Theta(\prod_{\ell<k} |A_\ell|)$, i.e., $O(m^p)$ per voter, making even the input exponentially large in the number of parties.  
To avoid this blow-up, we retain environmental voters in their simple form and restrict sequential reasoning to the designated strategic subset.

Instead of the best responses for a single voter,\footnote{
If all voters followed the simpler $(b_i, g_i)$ form, the game effectively collapses to the simultaneous case.  
Otherwise, one would need to impose additional succinctness restrictions on $f^p_{i,k}$, which would be ad hoc and fall outside our present scope.} we focus on \emph{dominant strategies} and \emph{equilibria among the strategic voters} to ensure that the input size remains polynomially bounded.

\paragraph{Equilibrium notion.}
We study pure subgame-perfect equilibria (SPNE) in the induced extensive-form game with observed prefixes; that is, strategy profiles $s$ that form a Nash equilibrium in every subgame rooted at any history $h_k \in C_{<k}$.

\begin{definition}[Sequential Dominant Strategy at Least $U$ (\textsc{SeqDS-$\ge U$})]
Given an instance as above, a designated player $i^\star$, and a threshold $U\in\mathbb{Q}$,
decide whether $i^\star$ has a strategy $s_{i^\star}$ that guarantees
$\ttutil_{i^\star}(s_{i^\star},s_{-i^\star}) \ge U$
for all admissible multi-stage strategies $s_{-i^\star}$ of her opponents.
\end{definition}

\begin{theorem}[\textsc{SeqDS-$\ge U$} is PSPACE-complete]\label{thm:seq-ds-pspace}
Under FPTP with fixed tie-breaking and polynomial-time GE evaluation, if the number of stages $p$ is part of the input (polynomially bounded horizon), then \textsc{SeqDS-$\ge U$} is PSPACE-complete.
For constant $p$, the problem lies in $\Sigma_{O(p)}^{\mathbf P}$.
\end{theorem}

\begin{proof}[Proof.]
\emph{Membership.}
The extensive form has depth $p$.
At a stage $k$, the joint action set has size $(m_k)^n$ (cf.\ Fig.~\ref{fig:tree} and the discussion:
each node can be indexed by the $n$-tuple of ballots $(a_{k1},\dots,a_{kn})$), so the branching factor
may be exponential in $n$, but depth-first evaluation needs only polynomial \emph{space}:
we store one path of length $p$ plus polynomial-time computable summaries at each node.
Utilities at leaves are obtained by computing the GE outcome and checking $\ttutil_{i^\star}\!\ge\!U$,
both in polynomial time by assumption. Hence the alternating (max–min) evaluation of the game tree
for the dominance predicate (``does $i^\star$ guarantee utility $\ge U$ against all opponents?'')
is decidable in PSPACE via standard DFS with backtracking.

\emph{Hardness (from QBF).}  
Let $\Phi = Q_1X_1\, Q_2X_2 \cdots Q_pX_p : \varphi(X)$ be a quantified 3-CNF formula.  
We construct a sequential-primaries game with $p$ stages encoding $\Phi$.

\textbf{Politicians.}  
For each variable $X_k$, create party~$k$ with nominees $\{x_k,\neg x_k\}$, and add two dummy parties $A_{\top}=\{x_\top\}$ and $A_{\bot}=\{x_\bot\}$.

\textbf{Control.}  
If $Q_k=\exists$, stage~$k$ is controlled by $i^\star$ (tie-breaking ensures her ballot decides);  
if $Q_k=\forall$, by an adversary $v_a$.  
Thus the nominee tuple $c=(c_1,\dots,c_p)$ encodes a full assignment to $X$.

\textbf{GE voters.}  
Two fixed blocs $B_\top,B_\bot$ (each of size $q{+}1$) always vote $x_\top$ and $x_\bot$, respectively.  
For each clause $C_t=(\ell_{t1}\vee\ell_{t2}\vee\ell_{t3})$, add a clause voter with policy  
\[
g_t(c)=\text{first of }[\ell_{t1},\ell_{t2},\ell_{t3},x_\bot,0]\text{ appearing among finalists.}
\]
If $C_t$ is falsified under $c$, this voter supports $x_\bot$; otherwise, a satisfied literal gains her vote.

\textbf{Utilities.}  
Let $U:=u_{i^\star}(x_\top)$, with $u_{i^\star}(x_\top)>$ $u_{i^\star}(x_\bot)>$ \\$u_{i^\star}(\text{others})$,  
so $\ttutil_{i^\star}\ge U$ iff $x_\top$ wins the GE.  

\textbf{Tally.}  
$x_\top$ has baseline $q{+}2$ votes ($B_\top$ plus $i^\star$);  
$x_\bot$ starts with $q{+}1$ ($B_\bot$) and gains one per falsified clause.  
With tie-breaking in favour of $x_\bot$, $x_\top$ wins iff all clauses are satisfied.

\textbf{Correctness.}  
At each stage:
- if $Q_k=\exists$, $i^\star$ chooses $c_k\in\{x_k,\neg x_k\}$ (OR-node);
- if $Q_k=\forall$, the adversary chooses (AND-node).  
At the leaf, $x_\top$ wins iff $\varphi$ is satisfied by the induced assignment.  
Thus $i^\star$ can guarantee $\ttutil_{i^\star}\!\ge\!U$ iff $\Phi$ is true.

The construction uses $O(p+q)$ voters and politicians and is polynomial in $|\Phi|$.  
Hence \textsc{SeqDS-$\ge U$} is PSPACE-hard. Combined with membership, it is PSPACE-complete.  
For fixed $p$, the alternation depth is $O(p)$, yielding the stated $\Sigma_{O(p)}^{\mathbf P}$ bound.
\end{proof}

The sequential dominance framework already captures the essential source of hardness: 
each stage introduces one quantifier alternation, mirroring the structure of a quantified Boolean formula.  
Equilibrium reasoning, however, does not fundamentally simplify this picture.  
In a turn-based extensive form, verifying subgame-perfection simply requires backward induction 
over the same alternating stages, while hardness stems from the same quantifier nesting.  
Hence the complexity of equilibrium existence grows in exactly the same way as dominance verification, 
culminating in the following theorem.

We reuse the \emph{environment/clause voters with cyclic fallback lists} and the two
top-level parties $\Aodd=\{a_1,a_3\}$, $\Aeven=\{a_2,a_4\}$ with fixed tie-breaking
$a_1\succ a_2\succ a_3\succ a_4$ exactly as in Theorem~\ref{thm:NE-exist}.
No new voter types are introduced; we only stage the same construction.

\begin{theorem}[SPNE existence reusing the NE gadget]\label{thm:SPNE-PSPACE}
Under FPTP with fixed tie-breaking and list-based GE behaviour,
deciding whether a pure subgame-perfect equilibrium (SPNE) exists is
$\Sigma_{O(p)}^{\mathbf P}$-complete for fixed number of stages $p$,
and PSPACE-complete when $p$ is part of the input.
\end{theorem}

\begin{proof}
\emph{Membership.} The same depth-first evaluation argument (cf.\ Fig.~\ref{fig:tree} and Theorem~\ref{thm:seq-ds-pspace}) applies. 
Given a polynomially bounded horizon and polynomial-time evaluable GE outcomes, verifying subgame-perfect optimality along the sequential tree requires only polynomial space. 
Hence \textsc{SPNE-Exist} is in PSPACE, and for constant~$p$ the induced alternation yields membership in $\Sigma_{O(p)}^{\mathbf P}$.

\emph{Hardness (from \textsc{QBF}).}
Let $\Psi=\exists X_1\,\forall Y_1\,\cdots\,\exists X_k\,\forall Y_k:\varphi(X,Y)$.
Create $p=2k$ stages in order $X_1,Y_1,\ldots,X_k,Y_k$.
For each variable $z\in X\cup Y$, make a party with nominees $\{z,\neg z\}$.
At stage $X_i$, $\vo$ chooses $x_i$ or $\neg x_i$; at stage $Y_i$, $\ve$ chooses $y_i$ or $\neg y_i$.
Keep the same background as in Theorem~\ref{thm:NE-exist}:
(i) $\Aodd$ has baseline GE margin $+1$; 
(ii) for each clause $C_t=(\ell_{t1}\vee\ell_{t2}\vee\ell_{t3})$, include clause voters $v_{t1},\dots,v_{t4}$ with lists
$L_{ti}=(\ell_{t1},\ell_{t2},\ell_{t3},a_i,a_{i+1},a_{i+2},a_{i+3})$ (indices mod~4),
who select the first finalist on their list.

\textbf{Correctness.}
(\(\Rightarrow\)) If $\Psi$ is true, there exists a \emph{history-contingent plan} for the $X$-stages:
for each $i$, a rule that maps the previously chosen $(Y_1,\ldots,Y_{i-1})$ to a choice of $x_i$ or $\neg x_i$
such that $\varphi$ holds for all $Y$-choices. Let $\vo$ follow this plan.
Then along any sequence of $Y$-moves, every clause is satisfied, clause voters never fall back, and the baseline $+1$
elects $\Aodd$ at the GE. Hence neither $\vo$ (already winning) nor $\ve$ (cannot induce a fallback) can profitably deviate
in any subgame; backward induction yields a pure SPNE.

(\(\Leftarrow\)) If $\Psi$ is false, then against any $X$-stage choices there exists a first $Y_i$ whose choice falsifies some clause.
At that subgame, the corresponding clause voters fall back cyclically, yielding the $3{:}1$ split that—together with the baseline—flips
the GE winner to $\Aeven$ (as in Thm.~\ref{thm:NE-exist}). Thus $\ve$ has a profitable deviation, contradicting subgame-perfection.
Therefore no pure SPNE exists.

Hence a pure SPNE exists iff $\Psi$ is true. The instance size is polynomial.
For fixed $p=2k$ this gives $\Sigma_{O(p)}^{\mathbf P}$-hardness; with unbounded $p$, PSPACE-hardness.
Combined with membership, the theorem follows.
\end{proof}

\noindent\textbf{Summary.}
The sequential model generalizes the simultaneous setting by introducing stagewise conditioning.  
When all primaries are resolved simultaneously—or sequentially but with ex-ante commitment to complete strategies, before any primary outcomes are observed—dominance and equilibrium verification involve only a single existential–universal alternation and thus remain within $\Sigma_2^{\mathbf P}$.  
Allowing voters to condition on earlier outcomes adds one quantifier block per stage, giving $\Sigma_{O(p)}^{\mathbf P}$ complexity for fixed $p$, and PSPACE-completeness once $p$ is part of the input.  
Together with Theorem~\ref{thm:NE-exist} and Theorem~\ref{thm:seq-ds-pspace}, this yields a unified view: sequential conditioning acts as a quantifier-expanding mechanism linking the polynomial hierarchy and PSPACE.

\section{Conclusion}

We presented the first formal characterization of the computational complexity of primary elections under both simultaneous and sequential settings.  
In the ex-ante (commitment) model, equilibrium and dominance verification involve one existential–universal alternation, yielding $\Sigma_2^{\mathbf P}$-completeness.  
Allowing voters to condition on earlier outcomes across $p$ stages introduces additional quantifier blocks, giving $\Sigma_{O(p)}^{\mathbf P}$ complexity for fixed $p$ and PSPACE-completeness once the number of stages becomes part of the input.

Beyond these classifications, the staged primary process emerges as a \emph{quantifier-expanding mechanism} that bridges the polynomial hierarchy and PSPACE.  
This unifies several equilibrium notions — dominant strategies, Nash equilibria, and subgame-perfect reasoning — under a common structural lens, showing how temporal sequencing amplifies computational difficulty even under simple voting rules.  
From a multi-agent perspective, this connects election design to the broader challenge of reasoning about sequential coordination among boundedly rational agents.

Our results suggest a form of \emph{complexity shield}: certain strategically undesirable behaviours are computationally hard to implement in the worst case.
At the same time, this hardness depends on parameters such as the number of parties and primary stages, which may be small in real elections and thus reduce complexity. Conversely, real voters often face richer information environments—such as observing primary tallies, forming beliefs about others’ behaviour, or dealing with timing uncertainty, which can further increase strategic complexity.
These issues are beyond the scope of this work but point to promising directions for future research.

\section{Future Work}

Our results only begin to capture the strategic and institutional richness of primary elections.  
Several directions remain open.

\paragraph{Richer strategic and institutional models.}  
Extending the framework to mixed strategies would illuminate probabilistic reasoning and equilibrium selection in competitive settings.  
Further, while our hardness results apply generally, exploring restricted domains, such as closed or open primaries, may uncover tractable subclasses.  
Alternative voting mechanisms, including run-off and ranked-choice systems, could lead to qualitatively different complexity boundaries.  
Some countries even integrate polling data into nomination outcomes, suggesting hybrid models that couple social choice with statistical inference.

\paragraph{Information, dynamics, and candidate behaviour.}  
Our sequential model assumes perfect information and static preferences.  
Relaxing these assumptions to incorporate noisy or incomplete information, belief updates, or adaptive learning would link primary elections to models of bounded rationality and dynamic games. Another promising direction is to introduce candidate positioning: when politicians choose spatial platforms strategically, primaries may reshape equilibrium positions and polarization patterns.

\paragraph{AI and simulation-based exploration.}  
Agent-based simulations offer a natural complement to our theoretical results, enabling empirical study of strategic adaptation under bounded rationality.  
Combining such simulations with reinforcement learning or evolutionary dynamics could reveal when complex equilibria are approximable in practice, linking computational hardness with observable collective behaviour.

Overall, this work takes a first step toward a unified theory of primaries within computational social choice, clarifying how sequential collective decisions reshape both strategic reasoning and the computational landscape of multi-agent systems.

%% file: appendix.tex
\section{Explosion of Strategy Expression of Full Extensive Form Game}\label{appx:extensive-form}

In $\S$~\ref{sec:framework}, we coarsely estimated the input size required to describe the full extensive-form strategy as $\Theta(|A|^{pn})$. 
We now provide a more detailed calculation.

Suppose the election scheme allows multiple participations, and that primaries occur sequentially. 
Within each primary, voters act simultaneously—thus, a voter cannot condition their choice on other voters’ actions within the same primary. 
However, the information revealed between primaries determines the structure of the extensive-form tree.

For each extensive-form strategy, one must specify a choice at every information set where the voter acts. 
Hence, the encoding space of a single strategy grows in proportion to the number of distinct decision nodes that voter may face.

We distinguish two informational assumptions:

\begin{description}
    \item[(Case 1)] \textbf{Voters observe all previous primary outcomes (perfect inter-stage information).}
    \item[(Case 2)] \textbf{Voters only observe the aggregate vote counts from previous primaries (coarse inter-stage information).}
\end{description}

\paragraph{Case 1.}
Without loss of generality, consider voter~$n$.
In the first primary (party~1), each voter has no prior information, and thus only one decision node.
Each voter’s action space is $A_1 \cup \{0\}$, where $0$ denotes abstention.
When voter~$n$ acts, there are $|m_1 + 1|^{n-1}$ possible configurations of the other voters’ actions.

Proceeding sequentially, suppose after party~$k$’s primary the current game tree has $\ell$ nodes.
In the next primary (party~$k{+}1$), each of the $n$ voters again chooses from $A_{k+1} \cup \{0\}$, yielding
$\ell \times |m_{k+1} + 1|^{n-1}$ new nodes.
By induction, the total number of decision nodes before the general election is
\[
\prod_{k=1}^{p} |m_k + 1|^{n-1} = \biggl(\prod_{k=1}^{p} |m_k + 1|\biggr)^{n-1}.
\]
By the AM--GM inequality,
\[
\prod_{k=1}^{p} |m_k + 1|
\;\leq\;
\biggl(\frac{\sum_{k=1}^{p} (m_k + 1)}{p}\biggr)^{p}
= \Bigl(\frac{m + p}{p}\Bigr)^{p},
\]
with equality when all $m_k$ are equal.
Hence the number of decision nodes is
\[
\Theta\!\left(\Bigl(\frac{m}{p}\Bigr)^{pn}\right).
\]
If $p$ is treated as a constant, this simplifies to $\Theta(|A|^{pn})$.

\paragraph{Case 2.}
Now assume voters only observe the aggregate \emph{vote counts} from previous primaries.
In this case, the number of distinct outcome profiles generated by other voters in party~$k$’s primary equals
\[
H^{n-1}_{m_k}
\;=\;
\binom{(n-1) + m_k - 1}{m_k},
\]
the number of \emph{combinations with repetition} (or multisets of size $m_k$ drawn from $n{-}1$ voters).
This quantity is asymptotically
\[
\Theta\bigl((n + m_k)^{m_k}\bigr)
\;=\;
\Theta\bigl(n^{m_k} + m_k^{m_k}\bigr).
\]
Therefore, the total number of decision nodes across all primaries is
\[
\Theta\!\left(\prod_{k=1}^{p} (n^{m_k} + m_k^{m_k})\right)
= \Theta\!\left(n^{m} + \prod_{k=1}^{p} m_k^{m_k}\right).
\]
Since $n > m$ and each $m_k \leq m$, the dominant term is $\Theta(n^{m})$.

%\medskip
In summary, although the explosion in input size arises from different sources—
party size and voter count in Case~1, versus total candidate count in Case~2—
both cases exhibit exponential growth in the parameters $(n, m, p)$.
In the computational problems considered in this paper, neither case yields a tractable representation once $p$ scales with input size.

\section{Case of No Equilibrium Even Without Participation Cost}\label{appx:no-kappa}

Relative to Example~\ref{EXM:NO-NE-Kappa}, we now  
(i) extend each party’s primary to three candidates,  
(ii) add a GE-only candidate~$a_7$ that both partisans strictly rank last, and  
(iii) adopt a fixed \emph{cyclic tie-breaking rule} within each primary, where ties are resolved along a pre-determined cycle (e.g., $a_1 \succ a_3 \succ a_5 \succ a_1$).  

The outsider~$a_7$ ensures that both parties share an outcome they wish to avoid, while the cyclic tie-breaking guarantees that in every primary configuration at least one partisan has a profitable deviation, thereby eliminating pure equilibria.

\begin{example}\label{EXM:NO-NE}
\textbf{Nonexistence of equilibrium under cyclic tie-breaking (no~$\kappa$).}

Consider FPTP with
\(A_{\mathrm{odd}}=\{a_1,a_3,a_5\}\),
\(A_{\mathrm{even}}=\{a_2,a_4,a_6\}\),
and a GE-only candidate \(a_7\).
Six cyclic voters \(v_1,\dots,v_6\) with list-form preference $L_i = [{a_{i},a_{i+1}, a_{i+2}, a_{1+3},a_{i+4},a_{i+5}}]  \text{ (mod~6)}$  and five fixed $a_7$-supporters never join primaries.
Two partisan voters remain strategic, each preferring their party’s nominees over the rival’s, and all over~$a_7$.

Within each primary, two-way ties are resolved by a fixed cycle:
\(a_1 \succ a_3 \succ a_5 \succ a_1\) for odd, and
\(a_2 \succ a_4 \succ a_6 \succ a_2\) for even.
As shown in Figure~\ref{fig:cyclic-voting}, whenever the finalists include
\((a_i, a_{i\pm3}, a_7)\), the outsider \(a_7\) wins. 
Both partisans strictly prefer to avoid this and can exploit the cyclic tie-breaking 
to alter nominees and improve their outcome. 
By symmetry, from every pure profile at least one partisan has a profitable deviation.
Hence no pure Nash equilibrium exists.
\end{example}

\begin{figure}[ht]
    \centering
       \includegraphics[width=0.25\textwidth, angle=0]{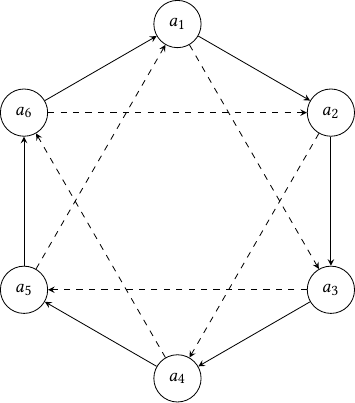}
    \caption{Illustration of Example~\ref{EXM:NO-NE}.  
    Solid arrows: if $c=(x,y,a_7)$, then $x$ wins the GE.  
    Dashed arrows: cyclic tie-breaking in primaries ($w\!\to\! z$ means $w$ defeats $z$ in a tie).}
    \Description{{ Illustration of Example~\ref{EXM:NO-NE}. Solid arrows indicate general-election outcomes: if $c = (x,y,a_7)$, then $x$ wins the GE.  Dashed arrows denote cyclic tie-breaking in primaries ($w{\rightarrow}z$ means $w$ defeats $z$ in a tie).}}
    \label{fig:cyclic-voting}
\end{figure}

%\begin{figure}[ht]
    %\centering
    %\resizebox{.8\linewidth}{!}{\input{Figures/Illustration_Completed_tree}}
    %\caption{Game tree of two voters in a sequential primary. Nodes labelled %\emph{1’s} and \emph{2’s} represent party~1’s and party~2’s primaries, respectively. At each node, both voters simultaneously choose ballots $(a_i,a_j)$.}
    %\Description{ The figure shows a horizontal game tree with three stages. The leftmost node, labelled “1’s”, represents Party 1’s primary election. From it, several arrows branch rightward, each labelled with a joint ballot such as $(a_1,a_1)$ or $(a_2, a_2)$. These arrows connect to three vertically stacked nodes, all labelled “2’s”, showing Party 2’s primary after different stage-1 outcomes.Between the 2’s nodes is a vertical ellipsis indicating additional possibilities. From the middle “2’s” node, a dashed continuation marked “...” leads rightward to a node labelled “GE” (general election).    The GE node branches into two terminal circles, labelled $(c_1,c_1)$ and $(c_p,c_p)$, with a vertical ellipsis between them to indicate further outcomes.The diagram thus depicts how the sequential primary process unfolds: Party 1 acts first, Party 2 responds based on observed history, and the process concludes in a general election.    }
    %\label{fig:tree}
%\end{figure}